\documentclass[11pt,letterpaper]{article}

\usepackage{amsmath,amsfonts,amsthm,amssymb,stmaryrd,relsize}
\theoremstyle{plain}

\numberwithin{equation}{section}

\newtheorem{thm}{Theorem}[section]
\newtheorem{lem}[thm]{Lemma}
\newtheorem{cor}[thm]{Corollary}

\newenvironment{exam}[1]
{\begin{flushleft}\textbf{Example #1}.\enspace}%
{\end{flushleft}}

\allowdisplaybreaks  

\newcommand{\complex}{{\mathbb C}}
\newcommand{\real}{{\mathbb R}}

\newcommand{\tbullet}{\raise .4ex\hbox{\tiny$\bullet$}} 

\newcommand{\iscripthat}{\widehat{\iscript}}
\newcommand{\jscripthat}{\widehat{\jscript}}
\newcommand{\cbar}{\overline{c}}

\newcommand{\rmtr}{\mathrm{tr\,}}
\newcommand{\rmin}{\mathrm{In\,}}
\newcommand{\rmid}{\mathrm{Id\,}}
\newcommand{\ityes}{\textit{yes}}
\newcommand{\itno}{\textit{no}}

\newcommand{\ascript}{\mathcal{A}}
\newcommand{\cscript}{\mathcal{C}}
\newcommand{\escript}{\mathcal{E}}
\newcommand{\iscript}{\mathcal{I}}
\newcommand{\jscript}{\mathcal{J}}
\newcommand{\kscript}{\mathcal{K}}
\newcommand{\lscript}{\mathcal{L}}
\newcommand{\mscript}{\mathcal{M}}
\newcommand{\oscript}{\mathcal{O}}
\newcommand{\pscript}{\mathcal{P}}
\newcommand{\sscript}{\mathcal{S}}

\newcommand{\ab}[1]{\left|#1\right|}
\newcommand{\doubleab}[1]{\left|\left|#1\right|\right|}
\newcommand{\brac}[1]{\left\{#1\right\}}
\newcommand{\paren}[1]{\left(#1\right)}
\newcommand{\sqbrac}[1]{\left[#1\right]}
\newcommand{\elbows}[1]{{\left\langle#1\right\rangle}}
\newcommand{\ket}[1]{{\left|#1\right>}}
\newcommand{\bra}[1]{{\left<#1\right|}}

\begin{document}

\title{FINITE QUANTUM INSTRUMENTS}
\author{Stan Gudder\\ Department of Mathematics\\
University of Denver\\ Denver, Colorado 80208\\
sgudder@du.edu}
\date{}
\maketitle

\begin{abstract}
This article considers quantum systems described by a finite-dimensional complex Hilbert space $H$. We first define the concept of a finite observable on $H$. We then discuss ways of combining observables in terms of convex combinations, post-processing and sequential products. We also define complementary and coexistent observables. We then introduce finite instruments and their related compatible observables. The previous combinations and relations for observables are extended to instruments and their properties are compared. We present four types of instruments; namely, identity, trivial, L\"uders and Kraus instruments. These types are used to illustrate different ways that instruments can act. We next consider joint probabilities for observables and instruments. The article concludes with a discussion of measurement models and the instruments they measure.
\end{abstract}

\section{Finite Observables}  
Let $\lscript (H)$ be the set of linear operators on a finite-dimensional complex Hilbert space $H$. For $S,T\in\lscript (H)$ we write $S\le T$ if $\elbows{\phi ,S\phi}\le\elbows{\phi ,T\phi}$ for all $\phi\in H$. We define the set of \textit{effects} by
\begin{equation*}
\escript (H)=\brac{a\in\lscript (H)\colon 0\le a\le 1}
\end{equation*}
where $0,1$ are the zero and identity operators, respectively. The effects correspond to \ityes -\itno\ experiments and
$a\in\escript (H)$ is said to \textit{occur} when a measurement of $a$ results in the value \ityes\ \cite{bgl95,hz12,kra83}.
A one-dimensional projection $P_\phi =\ket{\phi}\bra{\phi}$, where $\doubleab{\phi}=1$, is called an \textit{atom} and
$P_\phi\in\escript (H)$. We call $\rho\in\escript (H)$ a \textit{partial state} if $\rmtr (\rho )\le 1$ and $\rho$ is a \textit{state} if
$\rmtr (\rho )=1$. We denote the set of states by $\sscript (H)$ and the set of partial states by $\sscript _p(H)$. If
$\rho\in\sscript (H)$, $a\in\escript  (H)$ we call $\pscript _\rho (a)=\rmtr (\rho a)$ the \textit{probability that} $a$ \textit{occurs} in the state $\rho$ \cite{bgl95,hz12,kra83,nc00}.

We denote the unique positive square-root of $a\in\escript (H)$ by $a^{1/2}$. For $a,b\in\escript (H)$, their
\textit{sequential product} is the effect $a\circ b=a^{1/2}ba^{1/2}$ where $a^{1/2}ba^{1/2}$ is the usual operator product \cite{gg02,gn01,lud51}. We interpret $a\circ b$ as the effect that results from first measuring $a$ and then measuring $b$. We also call $a\circ b$ the effect $b$ \textit{conditioned on} the effect $a$ and write $(b\mid a)=a\circ b$. Notice that $\escript (H)$ is convex in the sense that if $b_i\in\escript (H)$ and $\lambda _i\ge 0$ with $\sum\limits _{i=1}^n\lambda _i=1$, then
$\sum\lambda _ib_i\in\escript (H)$ and
\begin{equation*}
\paren{\sum\lambda _ib_i\mid a}=\sum\lambda _i(b_i\mid a)
\end{equation*}
so $b\mapsto (b\mid a)$ is an affine function. In general, $a\mapsto (b\mid a)$ is not affine. Moreover, the product $\circ$ is not associative.

\begin{exam}{1}  
Let $a=\ket{\alpha}\bra{\alpha}$, $b=\ket{\beta}\bra{\beta}$ be atoms in $\escript (H)$. Then for any $c\in\escript (H)$ we have that
\begin{align*}
a\circ (b\circ c)&=\ket{\alpha}\bra{\alpha}\paren{\ket{\beta}\bra{\beta}c\ket{\beta}\bra{\beta}}\ket{\alpha}\bra{\alpha}\\
   &=\ab{\elbows{\alpha ,\beta}}^2\elbows{\beta ,c\beta}\ket{\alpha}\bra{\alpha}
\end{align*}
Moreover,
\begin{align*}
(a\circ b)\circ c&=\paren{\ket{\alpha}\bra{\alpha}\,\ket{\beta}\bra{\beta}\,\ket{\alpha}\bra{\alpha}}\circ c\\
   &=\ab{\elbows{\alpha ,\beta}}^2\ket{\alpha}\bra{\alpha}c\ket{\alpha}\bra{\alpha}\\
   &=\ab{\elbows{\alpha ,\beta}}^2\elbows{\alpha ,c\alpha}\ket{\alpha}\bra{\alpha}
\end{align*}
In general, $\elbows{\beta ,c\beta}\ne\elbows{\alpha ,c\alpha}$ so $a\circ (b\circ c)\ne (a\circ b)\circ c$.\hfill\qedsymbol
\end{exam}

Let $\Omega _A$ be a finite set. A \textit{finite observable with value-space} $\Omega _A$ is a subset
\begin{equation*}
A=\brac{A_x\colon x\in\Omega _A}\subseteq\escript (H)
\end{equation*}
such that $\sum\limits _{x\in\Omega _A}A_x=1$ \cite{bgl95,hz12,nc00}. We denote the set of finite observables on $H$ by
$\oscript (H)$. If $B=\brac{B_y\colon y\in\Omega _B}$ is another observable, we define the \textit{sequential product}
$A\circ B\in\oscript (H)$ to be the observable with value-space $\Omega _A\times\Omega _B$ given by
\begin{equation*}
A\circ B=\brac{A_x\circ B_y\colon (x,y)\in\Omega _A\times\Omega _B}
\end{equation*}
We also define the $A$-\textit{marginal} of $A\circ B$ to be the observable $(B\mid A)$ with value-space $\Omega _B$ given by
\begin{equation*}
(B\mid A)=\brac{(B\mid A)_y\colon y\in\Omega _B}\subseteq\escript (H)
\end{equation*}
where $(B\mid A)_y=\sum\limits _{x\in\Omega _A}(A_x\circ B_y)$. Since $\sum\limits _{y\in\Omega _B}(A_x\circ B_y)=A_x$ we say the $B$-\textit{marginal} of $A\circ B$ is $A$. We also call $(B\mid A)$ the observable $B$ \textit{conditioned} by the observable $A$ \cite{gud120,gud220}.

If $A\in\oscript (H)$ we define the effect-valued measure $X\mapsto A_X$ from $2^{\Omega _A}$ to $\escript (H)$ by
$A_X=\sum\limits _{x\in X}A_x$. By a slight misuse of terminology, we call $X\mapsto A_X$ an \textit{observable}. Moreover, we have the observable
\begin{equation*}
(B\mid A)_Y=\sum _{y\in Y}\sum _{x\in\Omega _A}(A_x\circ B_y)=\sum _{x\in\Omega _A}(A_x\circ B_Y)
\end{equation*}
and the observable
\begin{equation*}
(A\circ B)_\Delta =\sum _{(x,y)\in\Delta}(A_x\circ B_y)
\end{equation*}
In particular,
\begin{equation*}
(A\circ B)_{X\times Y}=\sum _{x\in X}(A_x\circ B_Y)
\end{equation*}
and we call $(A\circ B)_{X\times Y}$ the effect $(A_X\hbox{ then }B_Y)$ \cite{gud120,gud220}. It follows that
\begin{equation*}
(B\mid A)_Y=(A\circ B)_{\Omega _A\times Y}
\end{equation*}

If $\rho\in\sscript (H)$ and $A\in\oscript (H)$, the \textit{probability that} $A$ \textit{has a value in} $X\subseteq\Omega _A$, when the system is in state $\rho$ is $\pscript _\rho (A_X)=\rmtr(\rho A_X)$. Notice that $X\mapsto\pscript _\rho (A_X)$ is a probability measure on $\Omega _A$. We call
\begin{equation*}
\pscript _\rho (A_X\hbox{ then }B_Y)=\rmtr\sqbrac{\rho (A\circ B)_{X\times Y}}
\end{equation*}
the \textit{joint probability} of $A_X\hbox{ then }B_Y$. We now give alternative ways of writing this:
\begin{align*}
\pscript _\rho (A_X\hbox{ then }B_Y)&=\sum _{x\in X}\sum _{y\in Y}\rmtr (\rho A_x\circ B_y)
   =\sum _{x\in X}\rmtr (\rho A_x\circ B_Y)\\
   &=\rmtr\paren{\sum _{x\in X}A_x^{1/2}\rho A_x^{1/2}B_Y}=\rmtr\sqbrac{\sum _{x\in X}(A_x\circ\rho )B_Y}\\
   &=\rmtr\sqbrac{B_Y\circ\paren{\sum _{x\in X}(A_x\circ\rho )}}
\end{align*}

If $B^{(i)}\in\oscript (H)$ with the same value-space $\Omega$ and $\lambda _i\in\sqbrac{0,1}$, $i=1,2,\ldots ,n$, with
$\sum\lambda _i=1$ we can form the \textit{convex combination} observable $\sum\lambda _iB^{(i)}$ with
\begin{equation*}
\paren{\sum\lambda _iB^{(i)}}_y=\sum\lambda _iB_y^{(i)}
\end{equation*}
for all $y\in\Omega$ \cite{gud120,gud220,hz12}. We then have that
\begin{align*}
\paren{A\circ\sum\lambda _iB^{(i)}}_{(x,y)}&=A_x\circ\paren{\sum\lambda _iB^{(i)}}_y=\sum\lambda _iA_x\circ B_y^{(i)}\\
   &=\sum\lambda _i(A\circ B^{(i)})_{(x,y)}
\end{align*}
Hence,
\begin{equation*}
A\circ\sum\lambda _iB^{(i)}=\sum\lambda _iA\circ B^{(i)}
\end{equation*}
On the other hand
\begin{equation*}
\sqbrac{\sum\lambda _iA^{(i)}}\circ B\ne\sum\lambda _i(A^{(i)}\circ B)
\end{equation*}
in general. We also have that
\begin{align*}
\paren{\sum\lambda _iB^{(i)}\mid A}_y&=\sum _{x\in\Omega _A}\sqbrac{A_x\circ\paren{\sum\lambda _iB^{(i)}}_y}
   =\sum _{x\in\Omega _Aß}\sqbrac{A_x\circ\paren{\sum\lambda _iB_y^{(i)}}}\\
   &=\sum _{i=1}^n\lambda _i\sum _{x\in\Omega _A}A_x\circ B_y^{(i)}=\sum _{i=1}^n\lambda _i(B^{(i)}\mid A)_y
\end{align*}
Thus,
\begin{equation*}
\paren{\sum\lambda _iB^{(i)}\mid A}=\sum\lambda _i(B^{(i)}\mid A)
\end{equation*}
As before,
\begin{equation*}
\paren{B\mid\sum\lambda _iA^{(i)}}\ne\sum\lambda _i(B\mid A^{(i)})
\end{equation*}
in general.

Let $\nu =\sqbrac{\nu _{yz}}$ be a stochastic matrix with $\nu _{yz}\ge 0$, $\sum\limits _{z\in\Omega}\nu _{yz}=1$ for all
$y\in\Omega _B$ where $B\in\oscript (H)$. We define the observable $C=\nu\tbullet B$ with value-space $\Omega$ by
\begin{equation*}
(\nu\tbullet B)_z=\sum _{y\in\Omega _B}\nu _{yz}B_y
\end{equation*}
and call $C$ a \textit{post-processing} of $B$ \cite{gud120,gud220,hz12}. We then obtain
\begin{align*}
(A\circ\nu\tbullet B)_{(x,z)}&=A_x\circ (\nu\tbullet B)_z=A_x\circ\sum _{y\in\Omega _B}\nu _{yz}B_y\\
  &=\sum _{y\in\Omega _B}\nu _{yz}(A_x\circ B_y)
\end{align*}
Defining $\nu '_{\paren{(x,y),(x',z)}}=\nu _{yz}\delta _{xx'}$ we have that $\nu '$ is stochastic because
\begin{equation*}
\sum _{(x',z)}\nu _{\paren{(x,y),(x',z)}}=\sum _{(x',z)}\nu _{yz}\delta _{xx'}=\sum _z\nu _{yz}=1
\end{equation*}
Then
\begin{align*}
(A\circ\nu\tbullet B)_{(x',z)}&=\sum _{(x,y)\in\Omega _A\otimes\Omega _B}\nu _{\paren{(x,y),(x',z)}}(A_x\circ B_y)\\
   &=\sqbrac{\nu '\tbullet (A\circ B)}_{(x',z)}
\end{align*}
Hence, $A\circ\nu\tbullet B=\nu '(A\circ B)$ where $\nu '$ is essentially the same as $\nu$. Moreover,
\begin{align*}
(\nu\tbullet B\mid A)_z&=\sum _{x\in\Omega _A}A_x\circ (\nu\tbullet B)_z
   =\sum _{x\in\Omega _A}A_x\circ\paren{\sum _{y\in\Omega _B}\nu _{yz}B_y}\\
   &=\sum _{y\in\Omega _B}\nu _{yz}\sum _{x\in\Omega _A}A_x\circ B_y=\sum _{y\in\Omega _B}\nu _{yz}(B\mid A)_y
   =\sqbrac{\nu\tbullet (B\mid A)}_z
\end{align*}
so we conclude that $(\nu\tbullet B\mid A)=\nu\tbullet (B\mid A)$. Again, $(B\mid\nu\tbullet A)\ne\nu\tbullet (B\mid A)$ in general.

We now briefly consider three or more observables. For example, if $A,B,C\in\oscript (H)$, then
$A\circ (B\circ C)\in\oscript (H)$ is given by
\begin{equation*}
\sqbrac{A\circ (B\circ C)}_{(x,y,z)}=A_x\circ (B\circ C)_{(y,z)}=A_x\circ (B_y\circ C_z)
\end{equation*}
and for $\Delta\subseteq\Omega _A\times\Omega _B\times\Omega _C$ we define
\begin{equation*}
\sqbrac{A\circ (B\circ C)}_\Delta =\sum _{(x,y,z)\in\Delta}A_x\circ (B_y\circ C_z)
\end{equation*}
In particular, for $X\subseteq\Omega _A$, $Y\subseteq\Omega _B$, $Z\subseteq\Omega _C$ we have
\begin{align*}
\sqbrac{A\circ (B\circ C)}_{X\times Y\times Z}=\sum _{(x,y)\in X\times Y}A_x\circ (B_y\circ C_Z)\\
=\sum _{x\in X}\brac{A_x\circ\sqbrac{\sum _{y\in Y}(B_y\circ C_Z)}}
\end{align*}
Moreover, $\paren{(C\mid A)\mid B}\in\oscript (H)$ is given by
\begin{equation*}
\paren{(C\mid A)\mid B}_z=\sum _{y\in\Omega _B}\sqbrac{B_y\circ (C\mid A)_z}
  =\sum _{y\in\Omega _B}\sqbrac{B_y\circ\sum _{x\in\Omega _A}(A_x\circ C_z)}
\end{equation*}
and we have for every $Z\subseteq\Omega _C$ that
\begin{equation*}
\sqbrac{(C\mid A)\mid B}_Z=\sum _{y\in\Omega _B}\sqbrac{B_y\circ\sum _{x\in\Omega _A}(A_x\circ C_Z)}
\end{equation*}

We next discuss various types of observables. We call $B\in\oscript (H)$ an \textit{identity observable} if $B_y=\lambda _y1$ where $\lambda _y\in\real$ for every $y\in\Omega _B$. It follows that $\lambda _y\ge 0$ and
$\sum\limits _{y\in\Omega _B}\lambda _y=1$. Identity observables are the simplest types of observables. A convex combination of identity observables is an identity observable. Indeed, let $B^{(i)}$ be identity observables with
$B_y^{(i)}=\lambda _y^i1$, $y\in\Omega$, $i=1,2,\ldots ,n$. If $\mu_i\in\sqbrac{0,1}$ with $\sum\limits _{i=1}^n\mu _i=1$ we have that
\begin{equation*}
\sqbrac{\sum _{i=1}^n\mu _iB^{(i)}}_y=\sum _{i=1}^n\mu _iB_y^{(i)}=\sum _{i=1}^n\mu _i\lambda _y^i1
\end{equation*}
so $\sum\mu _iB^{(i)}$ is an identity observable. Also, if $B$ is an identity observable, the a post-processing of $B$ becomes
\begin{equation*}
(\nu\tbullet B)_z=\sum _{y\in\Omega _B}\nu _{yz}B_y=\sum _{y\in\Omega _B}\nu _{yz}\lambda _y1
\end{equation*}
and $\nu\tbullet B$ is an identity observable. If $A\in\oscript (H)$ and $B$ is an identity observable, we have that
\begin{align*}
(A\circ B)_{(x,y)}&=A_x\circ B_y=\lambda _yA_x
\intertext{and}
(B\mid A)_y&=\sum _{x\in\Omega _A}(A_x\circ B_y)=\sum _{x\in\Omega _A}\lambda _yA_x=\lambda _y1=B_y
\end{align*}
Of course, this latter property holds whenever $A$ and $B$ commute; that is $A_xB_y=B_yA_x$ for all $x\in\Omega _A$,
$y\in\Omega _B$. It also follows that if $A$ and $B$ are identity observables, then so is $A\circ B$.

An observable $A$ is \textit{atomic} if $A_x$ is an atom for all $x\in\Omega _A$. If $A$ is atomic, then
$A_x=\ket{\phi _x}\bra{\phi _x}$ and it follows that $\phi _x\perp\phi _{x'}$ for all $x\ne x'$. An $A\in\oscript (H)$ is \textit{indecomposable} if $A_x$ has rank~1 for all $x\in\Omega _x$. Clearly, an atomic observable is indecomposable but the converse does not hold. If
$A\in\oscript (H)$ is atomic with $A_x=\ket{\phi _x}\bra{\phi _x}$ and $B\in\oscript (H)$, we have
\begin{equation*}
(A\circ B)_{(x,y)}=\ket{\phi _x}\bra{\phi _x}B_y\ket{\phi _x}\bra{\phi _x}=\elbows{\phi _x,B_y\phi _x}P_{\phi _x}
\end{equation*}
so $A\circ B$ is indecomposable but is not atomic. It follows that if $A$ is indecomposable, then $A\circ B$ is also indecomposable for any $B\in\oscript (H)$. Moreover, we have that
\begin{equation*}
(B\mid A)_y=\sum _{x\in\Omega _A}(A_x\circ B_y)=\sum _{x\in\Omega _A}\elbows{\phi _x,B_y\phi _x}P_{\phi _x}
\end{equation*}
It follows that $(B\mid A)$ and $(C\mid A)$ commute for all $B,C\in\oscript (H)$.

We say that $A,B\in\oscript (H)$ \textit{coexist} if there exists an observable $C_{(x,y)}$ with value-space
$\Omega _A\times\Omega _B$ such that $A_x=\sum\limits _{y\in\Omega _B}C_{(x,y)}$ and
$B_y=\sum\limits _{x\in\Omega _A}C_{(x,y)}$ for all $x\in\Omega _A$, $y\in\Omega _B$ \cite{bgl95,hz12,kra83}.
Coexistence of $A$ and $B$ is interpreted as $A$ and $B$ being simultaneously measurable. We call $C_{(x,y)}$ a
\textit{joint observable} for $A$ and $B$. if $A$ and $B$ commute, then they coexist with joint observable $C_{(x,y)}=A_xB_y$. As a special case, we say that $a,b\in\escript (H)$ \textit{coexist} if there exist $a_1,b_1,c\in\escript (H)$ such that $a_1+b_1+c\le 1$ and $a=a_1+c$, $b=b_1+c$. The \textit{complement} of $a\in\escript (H)$ is defined as $a'=1-a$.

\begin{lem}    
\label{lem11}
{\rm\cite{hz12}}Two effects $a$ and $b$ coexist if and only if the observables $A=\brac{a,a'}$, $B=\brac{b,b'}$ exist.
\end{lem}
\begin{proof}
Suppose $a,b$ coexist and define $A_1=a$, $A_2=a'$, $B_1=b$, $B_2=b'$. Now there exist $a_1,b_1,c,d\in\escript (H)$ such that
$a_1+b_1+c+d=1$ and $a=a_1+c$, $b=b1+c$. Define the observable, $C_{(i,j)}$, $i,j=1,2$, by $C_{(1,2)}=c$,
$C_{(1,2)}=a_1$, $C_{(2,1)}=b_1$, $C_{(2,2)}=d$. Then
\begin{align*}
A_1&=a_1+c=C_{(1,2)}+C_{(1,2)}\\
A_2&=b_1+d=C_{(2,1)}+C_{(2,2)}\\
B_1&=b_1+c=C_{(1,1)}+C_{(2,1)}\\
B_2&=a_1+d=C_{(1,2)}+C_{(2,2)}
\end{align*}
Thus, the observables $A=\brac{A_1,A_2}$ and $B=\brac{B_1,B_2}$ coexist. Conversely, suppose $A$ and $B$ coexist. There there exists a joint observable $C_{(i,j)}$, $i,j=1,2$, such that $a=C_{(1,1)}+C_{(1,2)}$ and $b=C_{(1,1)}+C_{(2,1)}$. But then
\begin{equation*}
C_{(1,2)}+C_{(2,1)}+C_{(1,1)}\le 1
\end{equation*}
so $a$ and $b$ coexist.
\end{proof}

It is interesting to note that if $A,B\in\oscript (H)$, then $A$ and $(B\mid A)$ always coexist and a joint observable is
$A\circ B$. Indeed, we have that
\begin{align*}
\sum _x(A\circ B)_{(x,y)}&=\sum _x(A_x\circ B_y)=(B\mid A)_y
\intertext{and}
\sum _y(A\circ B)_{(x,y)}&=\sum _y(A_x\circ B_y)=A_x
\end{align*}

One can continue this discussion by saying that $A,B,C\in\oscript (H)$ \textit{coexist} if there exists a joint observable
$D_{(x,y,z)}$ with value-space $\Omega _A\times\Omega _B\times\Omega _C$ such that $A_x=\sum\limits _{y,z}D_{(x,y,z)}$, $B_y=\sum\limits _{x,z}D_{(x,y,z)}$, $C_z=\sum\limits _{x,y}D_{(x,y,z)}$. We then conclude that $A$, $(B\mid A)$ and
$\paren{(C\mid B)\mid A}$ coexist with joint observable
\begin{equation*}
D_{(x,y,z)}=A_x\circ (B_y\circ C_z)
\end{equation*}
Indeed,
\begin{align*}
\sum _{y,z}D_{(x,y,z)}&=\sum _{y,z}A_x\circ (B_y\circ C_z)=A_x\\
\sum _{x,z}D_{(x,y,z)}&=\sum _{x,z}A_x\circ (B_y\circ C_z)=\sum _x(A_x\circ B_y)=(B\mid A)_y\\
\sum _{x,y}D_{(x,y,z)}&=\sum _{z,y}A_x\circ (B_y\circ C_z)=\sum _xA_z\circ\sqbrac{\sum _y(B\circ C_z)}\\
   &=\sum _x\sqbrac{A_x\circ (C\mid B)_z}=\paren{(C\mid B)\mid A}_z
\end{align*}

Let $A,B\in\oscript (H)$ with $\ab{\Omega _A}=m$, $\ab{\Omega _B}=n$. We say that $A$ and $B$ are \textit{complementary} if
\begin{align*}
(B_y\mid A_x)&=A_x\circ B_y=\tfrac{1}{n}\,A_x\\
\intertext{and}
(A_x\mid B_y)&=B_y\circ A_x=\tfrac{1}{m}\,B_y
\end{align*}
for every $x\in\Omega _A$, $y\in\Omega _B$ \cite{gud120,gud220}. We interpret this as saying that when $A$ has a definite value $x$, then
$B$ is completely random and vice versa. A trivial example is when $A_x=\tfrac{1}{m}\,1$ and
$B_y=\tfrac{1}{n}\,1$ are completely random identity observables. When $A$ and $B$ are complementary we have that
\begin{equation*}
A\circ B=\brac{A_x\circ B_y\colon x\in\Omega _A,y\in\Omega _B}
   =\brac{\tfrac{1}{n}\,A_x,\cdots ,\tfrac{1}{n}\,A_x\colon x\in\Omega _A}
\end{equation*}
where there are $n$ terms $\tfrac{1}{n}\,A_x$ and we have a similar expression for $B\circ A$. Moreover,
\begin{equation*}
(B\mid A)_y=\sum _x(A_x\circ B_y)=\sum _x\tfrac{1}{n}\,A_x=\tfrac{1}{n}\,1
\end{equation*}
and similarly, $(A\mid B)_x=\tfrac{1}{m}\,1$. Thus, $(B\mid A)$ and $(A\mid B)$ are completely random identity observables. Moreover, we have that
\begin{align*}
\pscript _\rho (A_X\hbox{ then }B_Y)&=\sum _{x\in X}\sum _{y\in Y}\rmtr (\rho A_x\circ B_y)
     =\sum _{x\in X}\sum _{y\in Y}\rmtr\paren{\rho\,\tfrac{1}{n}\,A_x}\\
     &=\frac{\ab{Y}}{n}\,\rmtr (\rho A_X)=\frac{\ab{Y}}{n}\,\pscript _\rho (A_X)
\end{align*}
and similarly,
\begin{equation*}
\pscript _\rho (B_Y\hbox{ then }A_Z)=\frac{\ab{X}}{m}\,\pscript _\rho (B_Y)
\end{equation*}

We say that two orthonormal bases $\brac{\phi _i}$, $\brac{\psi _i}$ for $H$ are \textit{mutually unbiased} if
$\ab{\elbows{\phi _i,\psi _j}}^2=\tfrac{1}{n}$ for all $i,j=1,2,\ldots ,n$ \cite{wf89}. Mutually unbiased bases always exist
\cite{hz12,wf89}.

\begin{lem}    
\label{lem12}
Two atomic observables $A=\brac{P_{\phi _i}}$, $B=\brac{P_{\psi _i}}$ on $H$ are complementary if and only if $\brac{\phi _i}$ and $\brac{\psi _i}$ are mutually unbiased.
\end{lem}
\begin{proof}
We have that
\begin{equation*}
(B_j\mid A_i)=A_i\circ B_j=\ab{\elbows{\phi _i,\psi _j}}^2A_i
\end{equation*}
for $i,j=1,2,\ldots ,n$. Hence, $(B_j\mid A_i)=\tfrac{1}{n}\,A_i$ if and only if $\ab{\elbows{\phi _i,\psi _j}}^2=\tfrac{1}{n}$ for all $i,j=1,2,\ldots ,n$.
\end{proof}

\section{Finite Instruments}  
An \textit{operation} is an affine completely positive map $\ascript\colon\sscript _p(H)\to\sscript _p(H)\quad$
\cite{bgl95,hz12,nc00}. An operation $\ascript$ is a \textit{channel} if $\ascript (\rho )\in\sscript (H)$ for every
$\rho\in\sscript (H)$. We denote the set of channels on $H$ by $\cscript (H)$. Notice that if $a\in\escript (H)$, then
$\rho\mapsto (\rho\mid a)=a\circ\rho$ is an operation and if $A\in\oscript (H)$, then
$\rho\mapsto (\rho\mid A)=\sum\limits _x(A_x\circ\rho )$ is a channel. For a finite set $\Omega _\iscript$, a
\textit{finite instrument} with value-space $\Omega _\iscript$ is a set of operations
$\iscript =\brac{\iscript _x\colon x\in\Omega _\iscript}$ such that
$\iscripthat =\sum\limits _{x\in\Omega _\iscript}\iscript _x\in\cscript (H)$ \cite{bgl95,hz12,nc00}. Defining $\iscript _X$ for
$X\subseteq\Omega _\iscript$ by $\iscript _X=\sum\limits _{x\in X}\iscript _x$ we see that $X\mapsto\iscript _X$ is an operation-valued measure on $H$. If $A\in\oscript (H)$, we say that an instrument $\iscript$ is $A$-\textit{compatible} if
$\Omega _\iscript =\Omega _A$ and the \textit{probability reproducing condition}
\begin{equation}                
\label{eq21}
\pscript _\rho (A_X)=\rmtr\sqbrac{\iscript _X(\rho )}
\end{equation}
holds for every $\rho\in\sscript (H)$, $X\subseteq\Omega _A$ \cite{gud120,gud220,hz12}. To show that $\iscript$ is
$A$-compatible it is sufficient to show that $\pscript _\rho (A_x)=\rmtr\sqbrac{\iscript _x(\rho )}$ for every $\rho\in\sscript (H)$, $x\in\Omega _A$.

We view an $A$-compatible instrument as an apparatus that can be employed to measure the observable $A$. If $\iscript$ is an instrument, there exists a unique $A\in\oscript (H)$ such that $\iscript$ is $A$-compatible and we write $J(\iscript )=A$ \cite{hz12}. Then by \eqref{eq21} we have
\begin{equation*} 
\rmtr\sqbrac{\iscript _X(\rho )}=\rmtr\sqbrac{\rho J(\iscript )_X}
\end{equation*}
for every $\rho\in\sscript (H)$. We denote the set of instruments on $H$ by $\rmin (H)$. We show later that $J\colon\rmin (H)\to\oscript (H)$ is surjective but not injective. Thus, every $A\in\oscript (H)$ has many $A$-compatible instruments. We now give various examples of instruments.

Let $I(\rho )=\rho$ be the identity channel and $\Omega$ be a finite value-space. An \textit{identity instrument} $\rmid$ on
$\Omega$ has the form $\rmid _x=\lambda _xI$ where $\lambda _x\in\sqbrac{0,1}$,
$\sum\limits _{x\in\Omega}\lambda _x=1$. Thus, $\rmid _x(\rho )=\lambda _x\rho$ for all $\rho\in\sscript _p(H)$. Notice that $J(\rmid )$ is the identity observable $B_y=\lambda _y1$. If $A\in\oscript (H)$ and $\alpha\in\sscript (H)$ we define the \textit{trivial instrument} by $\iscript _x(\rho )=\rmtr (\rho A_x)\alpha$. Then $J(\iscript )=A$ and we conclude that $J$ is surjective.

If $A\in\oscript (H)$, we define the \textit{L\"uders instrument} $\lscript ^A$ by
\begin{equation*} 
\lscript _x^A(\rho )=(\rho\mid A)_x=A_x\circ\rho =A_x^{1/2}\rho A_x^{1/2}
\end{equation*}
for all $\rho\in\sscript _p(H)$ \cite{lud51}. Since $J(\lscript ^A)=A$, we see that $J$ is not injective. Notice that an identity instrument is a simple example of a L\"uders instrument. For another example, let $\brac{S_x\in\lscript (H)\colon x\in\Omega}$ satisfy $\sum S_x^*S_x=1$. Then $\rho\mapsto\sum\limits _{x\in\Omega}S_x\rho S_x^*$ is a channel and
$\kscript _x(\rho )=S_x\rho S_x^*$ gives an instrument called a \textit{Kraus instrument} with \textit{Kraus operators} $S_x$ \cite{kra83}. Notice that a L\'uders instrument $\lscript ^A$ is a Kraus instrument with operators $S_x=A_x^{1/2}$. Since
$\rmtr\sqbrac{\kscript _x(\rho )}=\rmtr (\rho S_x^*S_x)$ we see that $J(\kscript )=A$ where $A\in\oscript (H)$ is given by
$A_x=S_x^*S_x$. We define $K\colon\oscript (H)\to\rmin (H)$ by $K(A)=\lscript ^A$. We see that $K$ is not surjective because there are instruments like the trivial instruments and Kraus instruments that are not L\"uders type. Moreover, if
$\lscript ^A=\lscript ^B$, then when $\dim H=n$ we have that
\begin{equation*} 
\tfrac{1}{n}\,A_x=\lscript _x^A\paren{\tfrac{1}{n}\,1}=\lscript _x^B\paren{\tfrac{1}{n}\,1}=\tfrac{1}{n}\,B_x
\end{equation*}
Hence, $A_x=B_x$ so $K$ is injective.

\begin{exam}{2}  
This example shows that trivial instruments need not be Kraus instruments. Let $\iscript _X(\rho )=\rmtr (\rho A_x)P_\psi$ be a trivial instrument and suppose that $\iscript _x(\rho )=S_x\rho S_x^*$ is also a Kraus instrument. We then have that
\begin{equation*} 
P_\psi =\sum _x\iscript _x(\rho )=\sum _xS_x\rho S_x^*
\end{equation*}
for all $\rho\in\sscript (H)$. Letting $\rho =P_\phi$ we have that
\begin{equation*} 
P_\psi =\sum _x\ket{S_x\phi}\bra{S_x\phi}
\end{equation*}
for all $\phi\in H$ with $\doubleab{\phi}=1$. Let $\eta\in H$ with $\doubleab{\eta}=1$ satisfy $\eta\perp\psi$. Then
\begin{equation*} 
0=\elbows{\eta ,P_\psi\eta}=\sum _x\ab{\elbows{\eta ,S_x\phi}}^2
\end{equation*}
It follows that $\elbows{\eta ,S_x\phi}=0$ and hence, $S_x\phi =c_x\psi$ for $c_x\in\complex$. We conclude that
$S_x^*S_x=\ab{c_x}^2P_\psi$. But $\sum _xS_x^*S_x=1$ which gives a contradiction.\hfill\qedsymbol
\end{exam}

We have seen that if $\rmid _x=\lambda _xI$ is an identity instrument, then $J(\rmid )$ is the identity observable $B_x=\lambda _x1$. Conversely, if $B_x=\lambda _x1$, then
\begin{equation*} 
K(B)_x(\rho )=B_x^{1/2}\rho B_x^{1/2}=\lambda _x\rho =\rmid _x(\rho )
\end{equation*}
so $K(B)$ is an identity instrument. However, there are many other instruments that are $B$-compatible. In fact, we have that
$\iscript$ is $B$-compatible if and only if
\begin{equation}                
\label{eq22}
\iscript _x=\lambda _x\ascript _x
\end{equation}
where $\ascript _x\in\cscript (H)$ for all $x\in\Omega _\iscript$. Indeed, if $\iscript _x$ has the form \eqref{eq22}, then for all
$\rho\in\sscript (H)$ we have that
\begin{equation*} 
\rmtr\sqbrac{\iscript _x(\rho )}=\lambda _x\rmtr\sqbrac{\ascript _x(\rho )}=\lambda _x=\rmtr (\rho B_x)
\end{equation*}
Conversely, if $\rmtr\sqbrac{\iscript _x(\rho )}=\rmtr (\rho B_x)=\lambda _x$, for all $\rho\in\sscript (H)$, then when
$\lambda _x\ne 0$, we obtain
\begin{equation*} 
\rmtr\sqbrac{\tfrac{1}{\lambda _x}\,\iscript _x(\rho )}=1
\end{equation*}
Hence, $\tfrac{1}{\lambda _x}\,\iscript _X\in\cscript (H)$ so $\iscript _x=\lambda _x\ascript _x$ for some
$\ascript _x\in\cscript (H)$.

\begin{thm}    
\label{thm21}
{\rm{(i)}}\enspace $JK(A)=A$ for all $A\in\oscript (H)$.
{\rm{(ii)}}\enspace $KJ(\iscript )=\iscript$ if and only if $\iscript =K(A)$ for some $A\in\oscript (H)$.
\end{thm}
\begin{proof}
(i)\enspace For all $\rho\in\sscript (H)$ and $x\in\Omega _A$ we have that
\begin{equation*}
\rmtr\sqbrac{\rho JK(A)_x}=\rmtr\sqbrac{\rho\paren{J(\lscript ^A)_x}}=\rmtr (\rho A_x)
\end{equation*}
It follows that $JK(A)=A$.
(ii)\enspace If $\iscript =KJ(\iscript )$, then $\iscript =K(A)$ for $A=J(\iscript )$. Conversely, suppose $\iscript =K(A)$ for some $A\in\oscript (H)$. Then by (i) we have that
\begin{equation*}
KJ(\iscript )=KJ\paren{K(A)}=K\paren{JK(A)}=K(A)=\iscript\qedhere
\end{equation*}
\end{proof}

If $\iscript ^{(i)}\in\rmin (H)$, $i=1,2,\ldots ,n$, have the same value-space $\Omega$, $\lambda _i\in\sqbrac{0,1}$, with
$\sum\limits _{i=1}^n\lambda _i=1$, we define the \textit{convex combination}
$\sum\limits _{i=1}^n\lambda _i\iscript ^{(i)}\in\rmin (H)$ by
\begin{equation*} 
\sqbrac{\sum _{i=1}^n\lambda _i\iscript ^{(i)}}_x=\sum _{i=1}^n\lambda _i\iscript _x^{(i)}
\end{equation*}
for all $x\in\Omega$. The next result show that $J$ is affine, while $K$ is not.

\begin{thm}    
\label{thm22}
{\rm{(i)}}\enspace $J\sqbrac{\sum\lambda _i\iscript ^{(i)}}=\sum\lambda _iJ\sqbrac{\iscript ^{(i)}}$.
{\rm{(ii)}}\enspace $K\sqbrac{\sum\lambda _iA^{(i)}}\ne\sum\lambda _iK\sqbrac{A^{(i)}}$, in general.
\end{thm}
\begin{proof}
(i)\enspace For all $\rho\in\sscript (H)$ and $x\in\Omega$ we have that
\begin{align*}
\rmtr\sqbrac{\rho J\paren{\sum\lambda _i\iscript ^{(i)}}_x}&=\rmtr\sqbrac{\sum\lambda _i\iscript _x^{(i)}(\rho )}
  =\sum\lambda _i\rmtr\sqbrac{\iscript _x^{(i)}(\rho )}\\
  &=\sum\lambda _i\rmtr\sqbrac{\rho J(\iscript ^{(i)})_x}=\rmtr\sqbrac{\rho\sum\lambda _iJ(\iscript ^{(i)})_x}
\end{align*}
The result now follows.
(ii)\enspace For a counterexample, let $A,B\in\oscript (H)$ be atomic and let $\rho\in\sscript (H)$. In general, we have that
\begin{align*}
K\paren{\tfrac{1}{2}\,A+\tfrac{1}{2}\,B}_x(\rho )&=\lscript _x^{\frac{1}{2}\,A+\frac{1}{2}\,B}(\rho )
    =\tfrac{1}{2}(A_x+B_x)\rho (A_x+B_x)\\
    &=\tfrac{1}{2}(A_x\rho A_x+B_x\rho B_x+A_x\rho B_x+B_x\rho A_x)\\
    &\ne\tfrac{1}{2}(A_x\rho A_x+B_x\rho B_x)=\tfrac{1}{2}\,\lscript _x^A(\rho )+\tfrac{1}{2}\,\lscript _x^B(\rho )\\
    &=\tfrac{1}{2}\,K(A)_x(\rho )+\tfrac{1}{2}\,K(B)_x(\rho )
\end{align*}
so $K\paren{\tfrac{1}{2}\,A+\tfrac{1}{2}\,B}\ne\tfrac{1}{2}\,K(A)+\tfrac{1}{2}\,K(B)$.
\end{proof}

If $\nu =\sqbrac{\nu _{xy}}$ is a stochastic matrix and $\iscript\in\rmin (H)$, we define the \textit{post-processing} $\nu\tbullet\iscript\in\rmin (H)$ of $\iscript$ by
\begin{equation*} 
(\nu\tbullet\iscript )_y=\sum _{x\in\Omega _\iscript}\nu _{xy}\iscript _x
\end{equation*}
Notice that $\nu\tbullet\iscript\in\rmin (H)$ because
\begin{equation*}
(\nu\tbullet\iscript )_{\Omega _\iscript}=\sum _{x,y\in\Omega _\iscript}\nu _{xy}\iscript _x
   =\sum _{x\in\Omega _\iscript}\iscript _x\in\cscript (H)
\end{equation*}

\begin{thm}    
\label{thm23}
{\rm{(i)}}\enspace $J(\nu\tbullet\iscript )=\nu\tbullet J(\iscript )$.
{\rm{(ii)}}\enspace $K(\nu\tbullet A)\ne\nu\tbullet K(A)$, in general.
{\rm{(iii)}}\enspace $\nu\tbullet\sqbrac{\sum\lambda _i\iscript ^{(i)}}=\sum\lambda _i\nu\tbullet\iscript ^{(i)}$.
\end{thm}
\begin{proof}
(i)\enspace For every $\rho\in\sscript (H)$ and applicable $y$ we have that
\begin{align*}
\rmtr\sqbrac{\rho J(\nu\tbullet\iscript )_y}&=\rmtr\sqbrac{\rho J\paren{\sum _x\nu _{xy}\iscript _x}}
    =\rmtr\sqbrac{\rho\sum _x\nu _{xy}J(\iscript )_x}\\
    &=\rmtr\sqbrac{\rho\paren{\nu\tbullet J(\iscript )}_y}
\end{align*}
and the result follows.
(ii)\enspace For a counterexample, let $A\in\oscript (H)$ be atomic and let $\rho\in\sscript (H)$. In general, we obtain
\begin{align*}
K(\nu\tbullet A)_y(\rho )&=\lscript _y^{\nu\tbullet A}(\rho )=(\nu\tbullet A)_y^{1/2}\rho (\nu\tbullet A)^{1/2}\\
   &=\paren{\sum _x\nu _{xy}A_x}^{1/2}\rho\paren{\sum _{x'}\nu _{x'y}A_{x'}}^{1/2}\\
   &=\paren{\sum _x\nu _{xy}^{1/2}A_x}\rho\paren{\sum _{x'}\nu _{x'y}^{1/2}A_{x'}}\\
   &=\sum _{x,x'}\nu _{xy}^{1/2}\nu _{x'y}^{1/2}A_x\rho A_{x'}\\
   &\ne\sum _x\nu _{xy}A_x\rho A_x
\end{align*}
(iii)\enspace For every applicable $y$ we obtain
\begin{align*}
\sqbrac{\nu\tbullet\paren{\sum _i\lambda _i\iscript ^{(i)}}}_y&=\sum _x\nu _{xy}\paren{\sum _i\lambda _i\iscript ^{(i)}}_x
   =\sum _x\nu _{xy}\sum _i\lambda _i\iscript _x^{(i)}\\
   &=\sum _i\lambda _i\sum _x\nu _{xy}\iscript _x^{(i)}=\sum _i\lambda _i(\nu\tbullet\iscript ^{(i)})_y
\end{align*}
The result now follows.
\end{proof}

We say that $\iscript,\jscript\in\rmin (H)$ are \textit{complementary} if
\begin{equation*}
\rmtr\sqbrac{\jscript _y(J(\iscript )_x\circ\rho )}=\tfrac{1}{n}\,\rmtr\sqbrac{\iscript _x(\rho )}
\end{equation*}
for every $x\in\Omega _\iscript$, $y\in\Omega _\jscript$, $\rho\in\sscript (H)$ where $n=\ab{\Omega _\jscript}$ and
\begin{equation*}
\rmtr\sqbrac{\iscript _x(J(\jscript )_y\circ\rho )}=\tfrac{1}{m}\,\rmtr\sqbrac{\jscript _y(\rho )}
\end{equation*}
for every $x\in\Omega _\iscript$, $y\in\Omega _\jscript$, $\rho\in\sscript (H)$ where $m=\ab{\Omega _\iscript}$. As with observables, this says that when $\iscript$ has a definite value $x$, then $\jscript$ is completely random and vice versa.

\begin{lem}    
\label{lem24}
$\iscript$ and $\jscript$ are complementary if and only if $J(\iscript )$ and $J(\jscript )$ are complementary.
\end{lem}
\begin{proof}
The following statements are equivalent:
\begin{align*}
\rmtr\sqbrac{\jscript _y\paren{J(\iscript )_x\circ\rho}}&=\tfrac{1}{n}\,\rmtr\sqbrac{\iscript _x(\rho )}\\
\rmtr\sqbrac{J(\jscript )_yJ(\iscript )_x\circ\rho}&=\tfrac{1}{n}\,\rmtr\sqbrac{\iscript _X(\rho )}\\
\rmtr\sqbrac{\rho J(\iscript )_x^{1/2}J(\jscript )_yJ(\iscript )_x^{1/2}}&=\tfrac{1}{n}\,\rmtr\sqbrac{\rho J(\iscript )_x}\\
\rmtr\sqbrac{\rho J(\iscript )_x\circ J(\jscript )_y}&=\tfrac{1}{n}\,\rmtr\sqbrac{\rho J(\iscript )_x}
\end{align*}
This is equivalent to
\begin{equation*}
\paren{J(\jscript )_y\mid J(\iscript )_x}=J(\iscript )_x\circ J(\jscript )_y=\tfrac{1}{n}\,J(\iscript )_x
\end{equation*}
for all $x\in\Omega _\iscript$, $y\in\Omega _\jscript$. A similar expression holds with $\iscript$ and $\jscript$ interchanged so the result holds.
\end{proof}

\begin{cor}    
\label{cor25}
If $K(A)$ and $K(B)$ are complementary, then $A$ and $B$ are complementary.
\end{cor}
\begin{proof}
By Lemma~\ref{lem24} and Theorem~\ref{thm21}(i) we have that if $K(A)$ and $K(B)$ are complementary then $A=JK(A)$ and $B=JK(B)$ are complementary.
\end{proof}

We conjecture that the converse of Corollary~\ref{cor25} does not hold.

We say that $\iscript ,\jscript\in\rmin (H)$ \textit{coexist} if there exists a $\kscript\in\rmin (H)$ with value-space
$\Omega _\iscript\times\Omega _\jscript$ such that
\begin{equation}                
\label{eq23}
\iscript _x=\sum _{y\in\Omega _\jscript}\kscript _{(x,y)},\qquad\jscript _y=\sum _{x\in\Omega _\iscript}\kscript _{(x,y)}
\end{equation}

\begin{lem}    
\label{lem26}
{\rm{(i)}}\enspace If $\iscript$ and $\jscript$ coexist, then $J(\iscript )$ and $J(\jscript )$ coexist.
{\rm{(ii)}}\enspace If $K(A)$ and $K(B)$ coexist, then $A$ and $B$ coexist.
\end{lem}
\begin{proof}
(i)\enspace Since $\iscript$ and $\jscript$ coexist, there exists a $\kscript\in\rmin(H)$ satisfying \eqref{eq23}. Define the observable $C_{(x,y)}\in\oscript (H)$ with value-space $\Omega _\iscript\times\Omega _\jscript$ defined by
$C_{(x,y)}=J(\kscript _{(x,y)})$. For all $x\in\Omega _\iscript$, $\rho\in\sscript (H)$ we have that
\begin{align*}
\rmtr\sqbrac{\rho J(\iscript )_x}&=\rmtr\sqbrac{\iscript _x(\rho )}
    =\rmtr\sqbrac{\sum _{y\in\Omega _\jscript}\kscript _{(x,y)}(\rho )}\\
    &=\sum _{y\in\Omega _\jscript}\rmtr\sqbrac{\kscript _{(x,y)}(\rho )}
    =\sum _{y\in\Omega _\jscript}\rmtr\sqbrac{\rho J(\kscript _{(x,y)})}\\
    &=\rmtr\sqbrac{\rho\sum _{y\in\Omega _\jscript}J(\kscript _{(x,y)})}
    =\rmtr\sqbrac{\rho\sum _{y\in\Omega _\jscript}C_{(x,y)}}
\end{align*}

It follows that $J(\iscript )_x=\sum\limits _{y\in\Omega _\jscript}C_{(x,y)}$ and in a similar way,
$J(\jscript )_y=\sum\limits _{x\in\Omega _\iscript}C_{(x,y)}$. We conclude that $C_{(x,y)}$ is a joint observable for $J(\iscript )$ and $J(\jscript )$ so $J(\iscript )$ and $J(\jscript )$ coexist.
(ii)\enspace If $K(A)$ and $K(B)$ coexist then by Theorem~\ref{thm21}(i) we have that $A=JK(A)$ and $B=JK(B)$ coexist.
\end{proof}

Notice that if $\iscript ,\jscript\in\rmin (H)$ coexist, then $\iscripthat =\jscripthat$. The converse of this statement does not hold. Indeed, let $\iscript _x(\rho )=\rmtr (\rho A_x)\alpha$ and $\jscript _y(\rho )=\rmtr (\beta B_y)\alpha$ be trivial instruments. Then $\iscripthat =\jscripthat$ but if $A$ and $B$ do not coexist, then $\iscript$ and $\jscript$ do not exist. Also the converse of Theorem~\ref{lem26}(i) does not hold. Indeed, let $A,B\in\oscript (H)$ that coexist but for which $\lscript ^A\ne\lscript ^B$. Then $\lscript ^A$ and $\lscript ^B$ do not coexist. Hence, $J(\lscript ^A)$ and $J(\lscript ^B)$ coexist but $\lscript ^A$ and
$\lscript ^B$ do not. This also shows that the converse of Theorem~\ref{lem26}(ii) does not hold.

\section{Products of Instruments}  
For $\iscript ,\jscript\in\rmin (H)$, we define the \textit{product instrument} $\iscript\circ\jscript =\rmin (H)$ with value-space
$\Omega _\iscript\times\Omega _\jscript$ by
\begin{equation*}
(\iscript\circ\jscript )_{(x,y)}(\rho )=\jscript _y\sqbrac{\iscript _x(\rho )}
\end{equation*}
for all $\rho\in\sscript (H)$. We also define the \textit{conditioned instrument} with value-space $\Omega _\jscript$ by
\begin{equation*}
(\jscript\mid\iscript )_y=\sum _{x\in\Omega _\iscript}(\iscript\circ\jscript )_{(x,y)}
\end{equation*}
We then obtain
\begin{equation*}
(\jscript\mid\iscript )_y(\rho )=\sum _{x\in\Omega _\iscript}\jscript _y\sqbrac{\iscript _x(\rho )}
   =\jscript _y\sqbrac{\sum _{x\in\Omega _\iscript}\iscript _x(\rho )}=\jscript _y\sqbrac{\iscripthat (\rho )}
\end{equation*}

\begin{exam}{3}  
Let $\iscript _x(\rho )=S_x\rho S_x^*$, $\jscript _y (\rho )=T_y\rho T_y^*$ be Kraus instruments where
$\sum\limits _{x\in\Omega _\iscript}S_x^*S_x=\sum\limits _{y\in\Omega _\jscript}T_y^*T_y=1$. Then
\begin{equation*}
(\iscript\circ\jscript )_{(x,y)}(\rho )=T_yS_x\rho S_x^*T_y^*=T_yS_x\rho (T_yS_x)^*
\end{equation*}
We conclude that the product of two Kraus instruments is a Kraus instrument and its Kraus operators are $T_yS_x$. As we have seen, $J(\iscript )_x=S_x^*S_x$ and $J(\jscript )_y=T_y^*T_y$. Since
\begin{equation*}
\rmtr\sqbrac{(\iscript\circ\jscript )_{(x,y)}(\rho )}=\rmtr (T_yS_x\rho S_x^*T_y^*)=\rmtr (\rho S_x^*T_y^*T_yS_x)
\end{equation*}
we conclude that
\begin{equation*}
J(\iscript\circ\jscript )_{(x,y)}=S_x^*T_y^*T_yS_x
\end{equation*}
Now
\begin{align*}
\sqbrac{J(\iscript )\rho J(\jscript )}_{(x,y)}&=J(\iscript )_x\circ J(\jscript )_y=(S_x^*S_x)\circ (T_y^*T_y)\\
    &=(S_x^*S_x)^{1/2}T_y^*T_y(S_x^*S_x)^{1/2}
\end{align*}
which does not equal $S_x^*T_y^*T_yS_x$, in general. For example, let $S_x=\tfrac{1}{\sqrt{n\,}}U_x$ where $n=\dim H$ and $U_x$ is a unitary operator satisfying $U_xT_y^*T_y\ne T_y^*T_yU_x$. Then
\begin{equation*}
J(\iscript\circ\jscript )_{(x,y)}=\tfrac{1}{n}\,U_x^*T_y^*T_yU_x\ne\tfrac{1}{n}\,T_y^*T_y
   =\sqbrac{J(\iscript )\circ J(\jscript )}_{(x,y)}
\end{equation*}
Hence, $J(\iscript\circ\jscript )\ne J(\iscript )\circ J(\jscript )$, in general.

\hskip 2pc Moreover, we have that  
\begin{equation*}
(\jscript\mid\iscript )_y(\rho )=\sum _{x\in\Omega _\iscript}(\iscript\circ\jscript )_{(x,y)}(\rho )
    =T_y\sum _{x\in\Omega _\iscript}S_x\rho S_x^*T_y^*
\end{equation*}
Hence,
\begin{equation*}
\rmtr\sqbrac{(\jscript\mid\iscript )_y(\rho )}=\rmtr\sqbrac{\rho\sum _{x\in\Omega _\iscript}S_x^*T_y^*T_yS_x}
\end{equation*}
and it follows that
\begin{equation*}
J(\jscript\mid\iscript )_y=\sum _{x\in\Omega _\iscript}S_x^*T_y^*T_yS_x
\end{equation*}
As before,
\begin{align*}
\paren{J(\jscript)\mid J(\iscript )}_y&=\sum _{x\in\Omega _\iscript}\sqbrac{J(\iscript )_x\circ J(\jscript )_y}
     =\sum _{x\in\Omega _\iscript}(S_x^*S_x)\circ (T_y^*T_y)\\
     &\ne J(\jscript\mid\iscript )_y
\end{align*}
so $J(\jscript\mid\iscript )\ne\paren{J(\jscript )\mid J(\iscript )}$, in general.\hfill\qedsymbol
\end{exam}

\begin{exam}{4}  
For $A,B\in\oscript (H)$, let $\iscript =\lscript ^A$, $\jscript =\lscript ^B$ be their corresponding L\"uders instruments. Then $\iscript =K(A)$, $\jscript =K(B)$ and we have that
\begin{equation*}
(\iscript\circ\jscript )_{(x,y)}(\rho )=B_y\circ (A_x\circ\rho )
\end{equation*}
We conclude that
\begin{equation*}
\rmtr\sqbrac{(\iscript\circ\jscript )_{(x,y)}(\rho )}=\rmtr\sqbrac{\rho (A_x\circ B_y)}=\rmtr\sqbrac{\rho (A\circ B)_{(x,y)}}
\end{equation*}
Hence, $\iscript\circ\jscript$ is $A\circ B$ compatible and unlike general Kraus instruments, we have that
\begin{equation*}
J(\iscript\circ\jscript )=A\circ B=J(\iscript )\circ J(\jscript )
\end{equation*}
We also obtain
\begin{equation*}
\lscript _{(x,y)}^{A\circ B}(\rho )=(A\circ B)_{(x,y)}\circ\rho\ne B_y\circ (A_x\circ\rho )=(\iscript\circ\jscript )_{(x,y)}(\rho )
\end{equation*}
in general. Hence $K(A\circ B)\ne K(A)\circ K(B)$, in general. Moreover, it is not hard to show that
$K(A\circ B)=K(A)\circ K(B)$ if and only if $A$ and $B$ commute.

\hskip 2pc The conditioned instrument satisfies    
\begin{equation*}
(\jscript\mid\iscript )_y(\rho )=B_y\circ\sqbrac{\sum _{x\in\Omega _\iscript}(A_x\circ\rho )}
\end{equation*}
which gives
\begin{align*}
\rmtr\sqbrac{(\jscript\mid\iscript )_y(\rho )}&=\sum _{x\in\Omega _\iscript}\rmtr\sqbrac{B_y\circ (A_x\circ\rho )}
    =\sum _{x\in\Omega _\jscript}\rmtr\sqbrac{B_y(A_x\circ\rho )}\\
    &=\sum _{x\in\Omega _\iscript}\rmtr (\rho A_x\circ B_y)=\rmtr\sqbrac{\rho\sum _{x\in\Omega _\iscript}(A_x\circ B_y)}\\
    &=\rmtr\sqbrac{\rho (B\mid A)_y}
\end{align*}
Unlike general Kraus instruments we conclude that $J(\jscript\mid\iscript )=\paren{J(\jscript )\mid J(\iscript )}$.
\hfill\qedsymbol
\end{exam}

\begin{exam}{5}  
We now consider the identity instrument $\rmid _x=\lambda _xI$ where $\lambda _x\in\sqbrac{0,1}$,
$\sum _{x\in\Omega}\lambda _x=1$. If $\jscript\in\oscript (H)$ is arbitrary we have the following:
\begin{align*}
(\rmid\circ\jscript )_{(x,y)}(\rho )&=\jscript _y\sqbrac{\rmid _x(\rho )}=\jscript _y(\lambda _x\rho )
    =\lambda _x\jscript _y(\rho )\\
    (\jscript\circ\rmid )_{(x,y)}(\rho )&=\rmid _x\sqbrac{\jscript _y(\rho )}=\lambda _x\jscript _y(\rho )\\
    (\jscript\mid\rmid )_y(\rho )&=\sum _{x\in\Omega}\jscript _y\sqbrac{\rmid _x(\rho )}
    =\sum _{x\in\Omega}\jscript _y(\lambda _x\rho )=\jscript _y(\rho )\\
    (\rmid\mid\jscript )_x(\rho )&=\sum _{y\in\Omega _\jscript}\rmid _x\sqbrac{\jscript _y(\rho )}
    =\sum _{y\in\Omega _\jscript}\lambda _x\jscript _y(\rho )=\lambda _x\jscripthat (\rho )
\end{align*}
We conclude $(\jscript\mid\rmid )=\jscript$ and $(\rmid\mid\jscript )_x=\lambda _x\jscripthat$.\hfill\qedsymbol
\end{exam}

\begin{exam}{6}  
Let $\iscript _x(\rho )=\rmtr (\rho A_x)\alpha$ and $\jscript _y(\rho )=\rmtr (\rho B_y)\beta$ be trivial instruments. We have that
\begin{align*} 
(\iscript\circ\jscript )_{(x,y)}(\rho )&=\jscript _y\sqbrac{\iscript _x(\rho )}=\jscript _y\sqbrac{\rmtr (\rho A_x)\alpha}
    =\rmtr (\rho A_x)\jscript _y(\alpha )\\
    &=\rmtr (\rho A_x)\rmtr (\alpha B_y)\beta
\end{align*}
Hence,
\begin{align*}
\rmtr\sqbrac{(\iscript\circ\jscript )_{(x,y)}(\rho )}&=\rmtr (\rho A_x)\rmtr (\alpha B_y)\ne\rmtr (\rho A_x\circ B_y)\\
   &=\rmtr\sqbrac{\rho (A\circ B)_{(x,y)}}
\end{align*}
so we conclude that
\begin{equation*}
J(\iscript\circ\jscript )\ne A\circ B=J(\iscript )\circ J(\jscript )
\end{equation*}
in general. Moreover, in general we have that
\begin{align*}
\rmtr\sqbrac{(\jscript\mid\iscript )_y(\rho )}
&=\rmtr\sqbrac{\jscript _y\paren{\sum _{x\in\Omega _\iscript}\iscript _x(\rho )}}
   =\rmtr\sqbrac{\jscript _y(\alpha )}=\rmtr (\alpha B_y)\\
   &\ne\rmtr\sqbrac{\rho\sum _{x\in\Omega _\iscript}(A_x\circ B_y)}=\rmtr\sqbrac{\rho (B\mid A)_y}\\
   &=\rmtr\sqbrac{\rho\paren{J(\jscript )\mid J(\iscript )}_y}
\end{align*}
We conclude that $J(\jscript\mid\iscript )\ne\paren{J(\jscript )\mid J(\iscript )}$, in general.\hfill\qedsymbol
\end{exam}

For $\iscript ,\jscript\in\rmin (H)$ we define the \textit{joint probability}
\begin{equation*}
\pscript _\rho (\iscript _x\hbox{ then }\jscript _y)=\rmtr\sqbrac{\jscript _y\paren{\iscript _x(\rho )}}
\end{equation*}
for all $x\in\Omega _\iscript$, $y\in\Omega _\jscript$. For $X\subseteq\Omega _\iscript$, $Y\subseteq\Omega _\jscript$ we have that
\begin{align*}
\pscript _\rho (\iscript _X\hbox{ then }\jscript _Y)
    &=\sum _{x\in X}\sum _{y\in Y}\pscript _\rho (\iscript _x\hbox{ then }\jscript _y)\\
    &=\sum _{x\in X}\sum _{y\in Y}\rmtr\sqbrac{\jscript _y\paren{\iscript _x(\rho )}}
    =\rmtr\sqbrac{\jscript _Y\paren{\iscript _X(\rho )}}
\end{align*}
and we see this gives a probability measure $\Omega _\iscript\times\Omega _\jscript$.

\begin{exam}{7}  
Let $\iscript _x(\rho )=\rmtr (\rho A_x)\alpha$, $\jscript _y(\rho )=\rmtr (\rho B_y)\beta$ be trivial instruments. We then have that
\begin{align*}
\pscript _\rho (\iscript _x\hbox{ then }\jscript _y)&=\rmtr\sqbrac{\jscript _y\paren{\rmtr (\rho A_x)\alpha}}
   =\rmtr (\rho A_x)\rmtr (\alpha B_y)\\
\intertext{and}
\pscript _\rho (\iscript _X\hbox{ then }\jscript _Y)&=\rmtr (\rho A_X)\rmtr (\alpha B_Y)
\end{align*}
In general,
\begin{align*}
\rmtr (\rho A_x)\rmtr (\alpha B_y)&\ne\rmtr (\rho A_x\circ B_y)=\pscript _\rho (A_x\hbox{ then }B_y)\\
   &=\pscript _\rho\sqbrac{(J\iscript )_x\hbox{ then }(J\jscript )_y}
\end{align*}
We conclude that 
\begin{equation*}
\pscript _\rho (\iscript _x\hbox{ then }\jscript _y)\ne\pscript _\rho\sqbrac{(J\iscript )_x\hbox{ then }(J\jscript )_y}
\end{equation*}
in general.\hfill\qedsymbol
\end{exam}

\begin{lem}    
\label{lem31}
For any $A,B\in\oscript (H)$, $\rho\in\sscript (H)$ we have that
\begin{equation*}
\pscript _\rho\sqbrac{K(A)_X\hbox{ then }K(B)_Y}=\pscript _\rho (A_X\hbox{ then }B_Y)
\end{equation*}
\end{lem}
\begin{proof}
For $x\in\Omega _A$, $y\in\Omega _B$ we have that
\begin{align*}
\pscript _\rho\sqbrac{K(A)_x\hbox{ then }K(B)_y}&=\rmtr\sqbrac{K(B)_y\paren{K(A)_x(\rho )}}
  =\rmtr\sqbrac{B_y\circ (A_x\circ\rho )}\\
  &=\rmtr (\rho A_x\circ B_y)=\pscript _\rho (A_x\hbox{ then }B_y)
\end{align*}
The result now follows.
\end{proof}

\begin{exam}{8}  
If $\iscript =\lscript ^A$ and $\jscript =\lscript ^B$ are L\"uders instruments, we obtain
\begin{align*}
\pscript _\rho (\iscript _x\hbox{ then }\jscript _y)&=\rmtr\sqbrac{\lscript _y^B\paren{\lscript _x^A(\rho )}}=\rmtr\sqbrac{B_y\circ (A_x\circ\rho )}\\
   &=\rmtr\sqbrac{\rho (A_x\circ B_y)}=\pscript _\rho (A_x\hbox{ then }B_y)\\
   &\pscript _\rho\sqbrac{(J\iscript )_x\hbox{ then }(J\jscript )_y}
\end{align*}
Unlike the previous example we conclude that
\begin{equation*}
\pscript _\rho (\iscript _X\hbox{ then }\jscript _Y)=\pscript _\rho\sqbrac{(J\iscript )_X\hbox{ then }(J\iscript )_Y}
\end{equation*}
for all $\rho\in\sscript (H)$.\hfill\qedsymbol
\end{exam}

We close this section with a brief discussion on the instrument channel $\iscripthat =\iscript _{\Omega _\iscript}$ for an
$\iscript\in\rmin (H)$. It is well-known that any channel $\ascript$ has the form $\ascript (\rho )=\sum _xS_x\rho S_x^*$ where $S_x\in\lscript (H)$ and $\sum S_x^*S_x=1$ \cite{bgl95,hz12,kra83,nc00}. Although the operators $S_x$ are not unique, they are determined up to unitary equivalences \cite{hz12, nc00}. Corresponding to the channel
$\ascript (\rho )=\sum S_x\rho S_x^*$, we define the corresponding Kraus instrument
$\iscript _x^\ascript (\rho )=S_x\rho S_x^*$. We then have that $(\iscript ^\ascript )^\wedge =\ascript$. The instrument
$\iscript ^\ascript$ depends on the particular representation of $\ascript$ and is highly nonunique. Thus, there are other instruments
$\jscript\ne\iscript ^\ascript$ such that $\jscripthat =\ascript$. Of course, if $\iscript _x(\rho )=S_x\rho S_x^*$ is a Kraus instrument, then
$\ascript (\rho )=\sum _xS_x\rho S_x^*$ satisfies $\ascript =\iscripthat$ and $\iscript ^\ascript =\iscript$.

\begin{thm}    
\label{thm32}
$\iscript ^\ascript$ is an identity instrument if and only if $\ascript =I$.
\end{thm}
\begin{proof}
Let $\Omega$ be the value-space for $\iscript ^\ascript$. If $\iscript ^\ascript$ is an identity instrument, then $\iscript _x^\ascript (\rho )=\lambda _x\rho$ for every $x\in\Omega$, $\rho\in\sscript (H)$ where $\lambda _x\in\sqbrac{0,1}$ and
$\sum _{x\in\Omega}\lambda _x=1$. We then have that
\begin{equation*}
\ascript (\rho )=(\iscript ^\ascript )^\wedge (\rho )=\sum _{x\in\Omega}\iscript _x^\ascript (\rho )
    =\sum _{x\in\Omega}\lambda _x\rho =\rho
\end{equation*}
for every $\rho\in\sscript (H)$. Hence, $\ascript =I$. Conversely, suppose that $\ascript =I$, where $\ascript (\rho )=\sum _{x\in\Omega}S_x\rho S_x^*$. Letting $\rho = P_\phi$ be an atom we have that
\begin{equation*}
P_\phi =\ascript (P_\phi )=\sum _{x\in\Omega}S_xP_\phi S_x^*=\sum _{x\in\Omega}\ket{S_x\phi}\bra{S_x\phi}
\end{equation*}
Letting $\psi\in H$ with $\psi\perp\phi$ we obtain
\begin{equation*}
\sum _{x\in\Omega}\ab{\elbows{S_x\phi ,\psi}}^2=\sum _{x\in\Omega}\elbows{\psi ,S_x\phi}\elbows{S_x\phi ,\psi}
   =P_\phi\psi =0
\end{equation*}
Hence, $\elbows{S_x\phi ,\psi}=0$ for all $x\in\Omega$. It follows that $S_x\phi =c_x\phi$ for every $x\in\Omega$ where $c_x\in\complex$. Thus, $S_x=c_x1$ and we obtain
\begin{equation*}
\iscript _x^\ascript (\rho )=S_x\rho S_x^*=\ab{c_x}^2\rho
\end{equation*}
for all $\rho\in\sscript (H)$, $x\in\Omega$. We conclude that $\iscript ^\ascript$ is an identity instrument.
\end{proof}

\begin{cor}    
\label{cor33}
If $\iscript$ is an identity instrument, then $\iscripthat =I$.
\end{cor}

We conjecture that the converse of this corollary holds.
  
\begin{lem}    
\label{lem34}
For $\iscript ,\jscript\in\rmin (H)$ we have that
$(\iscript\circ\jscript )^\wedge =(\jscript\mid\iscript )^\wedge =\jscripthat\iscripthat$.
\end{lem}
\begin{proof}
For all $\rho\in\sscript (H)$ we have that
\begin{equation*}
(\iscript\circ\jscript )^\wedge (\rho )=(\iscript\circ\jscript )_{\Omega _\iscript\times\Omega _\jscript}(\rho )
   =\jscript _{\Omega _\jscript}\paren{\iscript _{\Omega _\iscript}(\rho )}=\jscripthat\paren{\iscripthat (\rho )}
\end{equation*}
Hence, $(\iscript\circ\jscript )^\wedge =\jscripthat\iscripthat$. Moreover, since
$(\jscript\mid\iscript )_y(\rho )=\jscript _y\paren{\iscripthat (\rho )}$ we obtain
\begin{equation*}
(\jscript\mid\iscript )^\wedge (\rho )=\sum _{y\in\Omega _\jscript}(\jscript\mid\iscript )_y(\rho )
    =\jscript _{\Omega _\jscript}\paren{\iscripthat (\rho )}=\jscripthat\paren{\iscripthat (\rho )}
\end{equation*}
and the result follows.
\end{proof}

\section{Measurement Models}  
A \textit{finite measurement model} (FIMM) is a 5-tuple $\mscript =(H,K,\eta ,\nu ,F)$ where $H,K$ are finite-dimensional Hilbert spaces called the \textit{base} and \textit{probe system}, respectively, $\eta\in\sscript (K)$ is an \textit{initial state},
$\nu\colon\sscript (H\otimes K)\to\sscript (H\otimes K)$ is a channel describing the measurement interaction between the base and probe systems and $F\in\oscript (K)$ is the \textit{pointer observable} \cite{bgl95,gud220,hz12}. The
\textit{model instrument} $\iscript ^\mscript\in\rmin (H)$ is the unique instrument given by
\begin{equation*}
\iscript _X^\mscript (\rho )=\rmtr _K\sqbrac{\nu (\rho\otimes\eta )(I\otimes F_X)}
\end{equation*}
where $\rmtr _K$ is the partial trace over $K$ \cite{bgl95,gud220,hz12} and the \textit{model observable} is
$A^\mscript =J(\iscript ^\mscript )$. We say that $\mscript$ \textit{measures} $\iscript ^\mscript$ and $A^\mscript$. The FIMM $\mscript$ is \textit{sharp} if $F$ is sharp.

An observable $A$ is \textit{commutative} if $A_xA_y=A_yA_x$ for all $x,y\in\Omega _A$ and $A,B\in\oscript (H)$
\textit{commute} if $A_xB_y=B_yA_x$ for all $x\in\Omega _A$, $y\in\Omega _B$. We say that two FIMMs
$\mscript _1,\mscript _2$ are \textit{simultaneous} if they are identical except for their probe observables $F,G$, respectively and $\mscript _1,\mscript _2$ \textit{commute} if $F,G$ commute.

\begin{thm}    
\label{thm41}
Two instruments $\iscript ,\jscript\in\rmin (H)$ coexist if and only if there exist simultaneous, commuting, sharp FIMMs
\begin{equation}                
\label{eq41}
\mscript _i=(H,K,\eta ,\nu ,F^i)\quad i=1,2
\end{equation}
such that $\mscript _1$ measures $\iscript$ and $\mscript _2$ measures $\jscript$.
\end{thm}
\begin{proof}
If $\iscript ,\jscript$ coexist, there is a joint instrument $\kscript\in\rmin (H)$. Now there exists a sharp FIMM
$\mscript =(H,K,\eta ,\nu ,F)$ \cite{bgl95,gud220,hz12} such that 
\begin{equation}                
\label{eq42}
\kscript _{(x,y)}(\rho )=\rmtr _K\sqbrac{\nu (\rho\otimes\eta )(I\otimes F_{(x,y)})}
\end{equation}
for all $(x,y)\in\Omega _\kscript$. Since $F$ is sharp, $F_{(x,y)}$ and $F_{(x',y')}$ are mutually orthogonal projections for all
$(x,y)\ne (x',y')$. Letting $F_x^1=\sum\limits _{y\in\Omega _\jscript}F_{(x,y)}$ and
$F_y^2=\sum\limits _{x\in\Omega _\iscript}F_{(x,y)}$ we have that $F_x^1$ and $F_y^2$ are mutually commuting projections and $F^1,F^2\in\oscript (K)$. We then obtain
\begin{align*}
\iscript _x(\rho )&=\sum _{y\in\Omega _\jscript}\kscript _{(x,y)}(\rho )=\rmtr _K\sqbrac{\nu (\rho\otimes\eta )(I\otimes F_x^1)}\\
\jscript _y(\rho )&=\sum _{x\in\Omega _\iscript}\kscript _{(x,y)}(\rho )=\rmtr _K\sqbrac{\nu (\rho\otimes\eta )(I\otimes F_y^2)}
\end{align*}
Hence, $\mscript _1$ and $\mscript _2$ given in \eqref{eq41} are simultaneous, commuting, sharp FIMMs that measure
$\iscript$ and $\jscript$. Conversely, suppose $\mscript _1$ and $\mscript _2$ given in \eqref{eq41} are simultaneous, commuting, sharp FIMMs that measure $\iscript$ and $\jscript$, respectively. Define the effects $F_{(x,y)}=F_x^1F_y^2$,
$(x,y)\in\Omega _\iscript\times\Omega _\jscript$. Then
\begin{equation*}
F=\brac{F_{(x,y)}\colon (x,y)\in\Omega _\iscript\times\Omega _\jscript}\in\oscript (K)
\end{equation*}
and $\mscript =(H,K,\eta ,\nu ,F)$ is a FIMM. We conclude that $\kscript _{(x,y)}$ given in \eqref{eq42} is an instrument on
$H$ and $\sum _{y\in\Omega _\jscript}\kscript _{(x,y)}=\iscript _x$,
$\sum _{x\in\Omega _\iscript}\kscript _{(x,y)}=\jscript _y$. Hence, $\iscript$ and $\jscript$ coexist.
\end{proof}

The analogous result for observables does not appear to follow directly from Theorem~\ref{thm41}, but we can use the following lemma.

\begin{lem}    
\label{lem42}
If $A,B\in\oscript (H)$ coexist, then there exists an $A$-compatible $\iscript\in\rmin (H)$ and a $B$-compatible
$\jscript\in\rmin (H)$ that coexist.
\end{lem}
\begin{proof}
Since $A,B\in\oscript (H)$ coexist, there exists a joint observable $C\in\oscript (H)$. Let $\iscript ,\jscript\in\rmin (H)$ be the trivial instruments $\iscript _x(\rho )=\rmtr (\rho A_x)\alpha$, $\jscript _y(\rho )=\rmtr (\rho B_y)\alpha$ for
$\alpha\in\sscript (H)$. Define $\kscript\in\rmin (H)$ by
\begin{equation*}
\kscript _{(x,y)}(\rho )=\rmtr\sqbrac{\rho C_{(x,y)}}\alpha
\end{equation*}
We then obtain
\begin{align*}
\sum _{y\in\Omega _\jscript}\kscript _{(x,y)}(\rho )&=\sum _{y\in\Omega _\jscript}\rmtr\sqbrac{\rho C_{(x,y)}}\alpha
    =\rmtr\paren{\rho\sum _{y\in\Omega _\jscript}C_{(x,y)}}\alpha\\
    &=\rmtr (\rho A_x)\alpha =\iscript _x(\rho )
\end{align*}
and similarly, $\sum\limits _{x\in\Omega _\jscript}\kscript _{(x,y)}(\rho )=\jscript _y(\rho )$ for all $\rho\in\sscript (H)$. Then
$\iscript$ is $A$-compatible, $\jscript$ is $B$-compatible and $\iscript$, $\jscript$ coexist.
\end{proof}

\begin{cor}    
\label{cor43}
Two observables $A,B\in\oscript (H)$ coexist if and only if there exist simultaneous, commuting, sharp FIMMs
\begin{equation*}
\mscript _i=(H,K,\eta ,\nu ,F^i)\quad i=1,2
\end{equation*}
\end{cor}
such that $\mscript _1$ measures $A$ and $\mscript _2$ measures $B$.
\begin{proof}
If $A,B$ coexist, by Lemma~\ref{lem42} there exists an $A$-compatible $\iscript\in\rmin (H)$ and a $B$-compatible
$\jscript\in\rmin (H)$ that coexist. By Theorem~\ref{thm41} there are simultaneous, commuting, sharp FIMMs given by
\eqref{eq41} such that $\mscript _1$ measures $\iscript$ and $\mscript _2$ measures $\jscript$. Since $\iscript$ is
$A$-compatible and $\jscript$ is $B$-compatible, $\mscript _1$ also measures $A$ and $\mscript _2$ also measures $B$. Conversely, suppose $\mscript _1$, $\mscript _2$ given by \eqref{eq41} are simultaneous, commuting, sharp FIMMs such that $\mscript _1$, $\mscript _2$ measure $A$, $B$, respectively. Defining $F_{(x,y)}=F_x^1F_y^2$ as in
Theorem~\ref{thm41}, $F\in\oscript (K)$ and the FIMM $\mscript =(H,K,\eta ,\nu ,F)$ measures a joint observable $C$ for $A$ and $B$ so $A$, $B$ coexist.
\end{proof}

Note that this work easily generalizes to three or more observables and instruments.

Consider a FIMM $\mscript =(H,K,\eta ,\nu ,F)$ for which $\nu =\nu _i\otimes\nu _2$ is factorized. Then 
\begin{align*}
\iscript _x(\rho )&=\rmtr _K\sqbrac{\nu _1\otimes\nu _2(\rho\otimes\eta )(I\otimes F_x)}
   =\rmtr _K\sqbrac{\nu _1(\rho )\otimes\nu _2(\eta )(I\otimes F_x)}\\
   &=\rmtr _K\sqbrac{\nu _1\rho )\otimes\nu _2(\eta )F_x}=\rmtr\sqbrac{\nu _2(\eta )F_x}\nu _1(\rho )
\end{align*}
This is not interesting because $\iscript$ is just a multiple of the channel $\nu _1$ and is similar to an identity instrument. The corresponding observable satisfies
\begin{equation*}
\rmtr (\rho A_x)=\rmtr\sqbrac{\iscript _x(\rho )}=\rmtr\sqbrac{\nu _2(\eta )F_x}
\end{equation*}
Hence, $A_x=\rmtr\sqbrac{\nu _2 (\eta )F_x}1$ is an identity observable which again is not interesting. We now consider a channel on $H\otimes K$ that is more general than a factorized channel but still has some factorized properties.

A \textit{von~Neumann operator} on $H\otimes K$ is a unitary operator $U$ for which there exists orthonormal basis
$\brac{\psi _i}$, $\brac{\phi _j}$ on $H,K$, respectively, such that
\begin{equation}                
\label{eq43}
U(\psi _i\otimes\phi _1)=\psi _i\otimes\phi _i\quad\hbox{ for all }i
\end{equation}
Condition \eqref{eq43} does not completely specify $U$ and there are many unitary operators that satisfy \eqref{eq43}. A completely defined unitary operator $U$ satisfying \eqref{eq43} is given by \eqref{eq43} and
\begin{align*}
U(\psi _i\otimes\phi _j)&=\psi _i\otimes\phi _j\quad\hbox{ if }i\ne j,j\ne 1,\\
U(\psi _i\otimes\phi _i)&=\psi _i\otimes\phi _1\quad\hbox{ if }i\ne 1
\end{align*}
A \textit{von~Neumann model} is a FIMM $\mscript =(H,K,P_{\phi _1},U,F)$ where $P_{\phi _1}$ is the initial pure state of the probe system and $U$ specifies a unitary channel $\nu (\rho ')=U\rho 'U^*$ for all $\rho'\in\sscript (H\otimes K)$, that is given by a von~Neumann operator $U$ \cite{hz12}.

\begin{thm}    
\label{thm44}
If $\mscript =(H,K,P_{\phi _1},U,F)$ is a von~Neumann model, then the measured instrument is
\begin{equation*}
\iscript _X(\rho )=\sum _{i,j}\elbows{\psi _i\otimes\phi _j,\rho\otimes F_X\psi _j\otimes\phi _i}\ket{\psi _i}\bra{\psi _j}
\end{equation*}
The corresponding instrument channel is
\begin{equation*}
\iscripthat (\rho )=\sum _iP_{\psi _i}\rho P_{\psi _i}
\end{equation*}
and the measured observable is
\begin{equation}                
\label{eq44}
A_X=\sum _i\elbows{\phi _i,F_X\phi _i}P_{\psi _i}
\end{equation}
\end{thm}
\begin{proof}
If $P_\psi$ is a pure state on $H$, then
\begin{align*}
U(P_\psi\otimes P_{\phi _1})U^*&=UP_{\psi\otimes\phi _1}U^*=U\paren{\ket{\psi\otimes\phi _1}\bra{\psi\otimes\phi _1}}U^*\\
    &=\ket{U(\psi\otimes\phi _1)}\bra{U(\psi\otimes\phi _1)}=P_{U(\psi\otimes\phi _1)}
\end{align*}
Now $\psi =\sum c_i\psi _i$ where $c_i=\elbows{\psi _i,\psi}$ so we have that
\begin{align*}
U(P_\psi\otimes P_{\phi _1})U^*
   &=\ket{U\paren{\sum _ic_i\psi _i\otimes\phi _1}}\bra{U\paren{\sum _jc_j\psi _j\otimes\phi _1}}\\
   &=\sum _{i,j}c_i\cbar _j\ket{U(\psi _i\otimes\phi _1)}\bra{U(\psi _j\otimes\phi _1)}\\
   &=\sum _{i,j}c_i\cbar _j\ket{\psi _i\otimes\phi _i}\bra{\psi _j\otimes\phi _j}\\
   &=\sum _{i,j}c_i\cbar _j\ket{\psi _i}\bra{\psi _j}\otimes\ket{\phi _i}\bra{\phi _j}
\end{align*}
Hence,
\begin{align*}
\iscript _X(P_\psi )&=\rmtr _K\sqbrac{U(P_\psi\otimes P_{\phi _1})U^*(I\otimes F_X)}\\
    &=\sum _{i,j}c_i\cbar _j\rmtr _K\sqbrac{\ket{\psi _i}\bra{\psi _j}\otimes\ket{\phi _i}\bra{\phi _j}F_X}\\
    &=\sum _{i,j}c_i\cbar _j\sqbrac{\rmtr\ket{\phi _i}\bra{\phi _j}F_X}\ket{\psi_i}\bra{\psi _j}\\
    &=\sum _{i,j}c_i\cbar _j\elbows{\phi _j,F_X\phi _i}\ket{\psi _i}\bra{\psi _j}\\
    &=\sum _{i,j}\elbows{\psi _i,\psi}\elbows{\psi ,\psi _j}\elbows{\phi _j,F_X\phi _i}\ket{\psi _i}\bra{\psi _j}\\
    &=\sum _{i,j}\elbows{\psi _i,P_\psi\psi _j}\elbows{\phi _j,F_X\phi _i}\ket{\psi _i}\bra{\psi _j}
\end{align*}
It follows that if $\rho\in\sscript (H)$, then
\begin{align*}
\iscript _X(\rho )&=\sum _{i,j}\elbows{\psi _i,\rho\psi _j}\elbows{\phi _j,F_X\phi _i}\ket{\psi _i}\bra{\psi _j}\\
   &=\sum _{i,j}\elbows{\psi _i\otimes\phi _j,\rho\otimes F_X\psi _j\otimes\phi _i}\ket{\psi _i}\bra{\psi _j}
\end{align*}
The instrument channel becomes
\begin{equation*}
\iscripthat (\rho )=\iscript _\Omega (\rho )=\sum _i\elbows{\psi _i,\rho\psi _i}\ket{\psi _i}\bra{\psi _i}
    =\sum _iP_{\psi _i}\rho P_{\psi _i}
\end{equation*}
The measured observable $A$ satisfies
\begin{equation*}
\elbows{\psi ,A_X\psi}=\rmtr (P_\psi A_X)=\rmtr\sqbrac{\iscript _X(P_\psi )}
    =\sum _i\ab{\elbows{\psi ,\psi _i}}^2\elbows{\phi _i,F_X\phi _i}
\end{equation*}
and it follows that
\begin{equation*}
A_X=\sum _i\elbows{\phi _i,F_X\phi _i}P_{\psi _i}
\qedhere
\end{equation*}
\end{proof}

\begin{cor}    
\label{cor45}
An observable $A=\brac{A_x\colon x\in\Omega}$ is measured by a von~Neumann model if and only if $A$ is commutative.
\end{cor}
\begin{proof}
If $A$ is measured by a von~Neumann model, then we see by \eqref{eq44} that $A$ is commutative. Conversely, suppose
$A$ is commutative. Then the $A_x$, $x\in\Omega$ are simultaneously diagonalizable. Hence, there is a basis
$\brac{\psi _i}$ of $H$ such that
\begin{equation*}
A_x=\sum _i\elbows{\psi _i,A_x\psi _i}P_{\psi _i}
\end{equation*}
for all $x\in\Omega$. Let $K$ be a Hilbert space with $\dim K=\dim H$ and let $\brac{\phi _j}$ be an orthonormal basis for
$K$. We set $\phi _1$ to be the initial probe state and $U$ to be the von~Neumann channel corresponding to
$\brac{\psi _i}$, $\brac{\phi _j}$. Define the probe observable $F$ by
\begin{equation*}
F_x=\sum _j\elbows{\psi _j,A_x\psi _j}P_{\phi _j}
\end{equation*}
Then $\elbows{\phi _i,F_x\phi _i}=\elbows{\psi _i,A_x\psi _i}$ so \eqref{eq44} holds. Hence $A$ is measured by the von~Neumann model $(H,K,P_\phi ,U,F)$.
\end{proof}

A \textit{normal} FIMM has the form $(H,K,P_\phi ,U,F)$ where $P_\phi$ is a pure state, $U$ represents a unitary channel and $F$ is an atomic pointer observable \cite{hz12}.

\begin{thm}    
\label{thm46}
$\iscript\in\rmin (H)$ is a Kraus instrument if and only if $\iscript$ is measured by a normal FIMM.
\end{thm}
\begin{proof}
Suppose that $\iscript$ is measured by the normal FIMM$=(H,K,P_\phi ,U,F)$. Letting $\brac{\psi _i}$, $\brac{\phi _i}$ be orthonormal bases for $H,K$, respectively, we employ the following formula for $\rmtr _K$ \cite{hz12}.
\begin{equation*}
\rmtr _K(T)=\sum _{j,k,n}\elbows{\psi _j\otimes\phi _k,T\psi _n\otimes\phi _k}\ket{\psi _j}\bra{\psi _n}
\end{equation*}
Letting $F_x=\ket{\phi _x}\bra{\phi _x}$ and $\psi\in H$, we obtain
\begin{align*}
\iscript _x(P_\psi )&=\rmtr _K\sqbrac{U(P_\psi\otimes P_\phi )U^*(I\otimes F_x)}\\
   &=\sum _{j,k,n}\elbows{\psi _j\otimes\phi _kUP_{\psi\otimes\phi}U^*%
      \paren{I\otimes\ket{\phi _x}\bra{\phi _x}}\psi _n\otimes\phi _k}\ket{\psi _j}\bra{\psi _n}\\
      &=\sum _{j,k,n}\elbows{\psi _j\otimes\phi _k,UP_{\psi\otimes\phi}U^*%
      \paren{\psi _n\otimes\elbows{\phi  _x,\phi _k}\phi _x}}\ket{\psi _j}\bra{\psi _n}\\  
      &=\sum _{j,n}\elbows{\psi _j\otimes\phi _x,UP_{\psi\otimes\phi}U^*(\psi _n\otimes\phi _x)}\ket{\psi _j}\bra{\psi _n}\\
      &=\sum _{j,n}\elbows{\psi _j\otimes\phi _x,U\elbows{\psi\otimes\phi ,U^*(\psi _n\otimes\phi _x)}\psi\otimes\phi}
         \ket{\psi _j}\bra{\psi _n}\\
      &=\sum _{j,n}\elbows{U(\psi\otimes\phi ),\psi _n\otimes\phi _x}\elbows{\psi _j\otimes\phi _x,U(\psi\otimes\phi )}
        \ket{\psi _j}\bra{\psi _n}      
\end{align*}
Define the operator $S_x$ on $H$ by
\begin{equation}                
\label{eq45}
S_x\psi =\sum _j\elbows{\psi _j\otimes\phi _x,U(\psi\otimes\phi)}\psi _j
\end{equation}
We then have that
\begin{equation*}
\iscript _x(P_\psi )=\ket{S_x\psi}\bra{S_x\psi}=S_xP_\psi S_x^*
\end{equation*}
It follows that $\iscript _x(\rho )=S_x\rho S_x^*$ for all $x\in\Omega _\iscript$, $\rho\in\sscript (H)$. Hence, $\iscript$ is a Kraus instrument.

Conversely, if $\iscript\in\rmin (H)$ is a Kraus instrument with $\iscript _x(\rho )=S_x\rho S_x^*$ then by Ozawa's Theorem,
$\iscript$ is measured by a FIMM $\mscript =(H,K,P_\phi ,U,F)$ where $F$ is sharp. If $F$ is not atomic, there is an $x$ such that $F_x$ is not an atom. For simplicity we can assume that $F_x=P_\alpha +P_\beta$. By our previous work, we conclude that there exists operators $S_1,S_2$ on $H$ such that
\begin{equation*}
\iscript _x(\rho )=S_1\rho S_1^*+S_2PS_2^*
\end{equation*}
But this contradicts the fact that $\iscript$ is a Kraus instrument. Hence, $F$ is atomic, so $\mscript$ is normal.
\end{proof}

\begin{cor}    
\label{cor47}
$\iscript\in\rmin (H)$ is a L\"uders instrument if and only if $\iscript$ is measured by a normal FIMM
$\mscript =(H,K,P_\phi ,U,F)$ that satisfies
\begin{equation}                
\label{eq46}
\elbows{\psi\otimes\phi _x,U(\psi\otimes\phi )}\ge 0
\end{equation}
for all $\psi\in H$, $x\in\Omega _\iscript$.
\end{cor}
\begin{proof}
Letting $S_x$ be the operator given by \eqref{eq45} we have that
\begin{align*}
\elbows{\psi ,S_x\psi}&=\sum _j\elbows{\psi _j\otimes\phi _x,U(\psi\otimes\phi )}\elbows{\psi ,\psi _j}\\
   &\elbows{\psi\otimes\phi _x,U(\psi\otimes\phi )}
\end{align*}
Hence, \eqref{eq46} holds if and only if $S_x\ge 0$ for all $x\in\Omega _\iscript$. When this is the case we have that
$\iscript _x(\rho )=A_x^{1/2}\rho A_x^{1/2}$ for the effect $A_x=S_x^2$ which is equivalent to $\iscript$ being a L\"uders instrument.
\end{proof}

A unitary operator $U\colon H\otimes H\to H\otimes H$ that satisfies $U(\psi\otimes\phi )=\phi\otimes\psi$ is called a
\textit{swap} operator. It is easy to show that a swap operator $U$ satisfies $U(\rho\otimes\eta )U^*=\eta\otimes\rho$ for all
$\eta ,\rho\in\sscript (H)$. A \textit{trivial} FIMM has the form $\mscript =(H,H,\eta ,U,F)$ where $U$ is a swap operator.

\begin{thm}    
\label{thm48}
An instrument $\iscript$ is trivial if and only if $\iscript$ is measured by a trivial {\rm FIMM}.
\end{thm}
\begin{proof}
Suppose $\iscript$ is measured by a trivial FIMM $\mscript =(H,H,\eta ,U,F)$. Letting $\rmtr _H$ be the partial trace over the second Hilbert space, we have that
\begin{align*}
\iscript _x(\rho )&=\rmtr _H\sqbrac{U(\rho\otimes\eta )U^*(I_1\otimes F_x)}
     =\rmtr _H\sqbrac{(\eta\otimes\rho )(I_1\otimes F_x)}\\
     &=\rmtr _H(\eta\otimes\rho F_x)=\rmtr (\rho F_x)\eta
\end{align*}
Then $\iscript$ is trivial. Conversely, let $\iscript\in\rmin (H)$ be trivial with $\iscript _x(\rho )=\rmtr (\rho A_x)\alpha$. For the trivial FIMM $\mscript =(H,H,\alpha ,U,A)$, our previous calculation shows that
\begin{equation*}
\rmtr _H\sqbrac{U(\rho\otimes\alpha )U^*(I_1\otimes A_x)}=\iscript _x(\rho )
\end{equation*}
We conclude that $\iscript$ is measured by $\mscript$.
\end{proof}

We have show that when $\mscript =(H,H,\eta ,U,F)$ is trivial, then the instrument measured by $\mscript$ has observable $F$ and state $\eta$.

\end{document}